\RequirePackage{fix-cm}
\documentclass[smallextended]{svjour3}       % onecolumn (second format)
\smartqed  % flush right qed marks, e.g. at end of proof
\usepackage{graphicx}
\usepackage{amsfonts}
\usepackage{amsmath}

\begin{document}

\title{Nonlocal sets of orthogonal multipartite product states with less members}
\author{Hui-Juan Zuo$^*$\and Jia-Huan Liu\and Xiao-Fan Zhen\and Shao-Ming Fei}
\institute{H.-J. Zuo \and J.-H. Liu\and X.-F. Zhen \at
              School of Mathematical Sciences, Hebei Normal University, Shijiazhuang, 050024, China \\
              \email{huijuanzuo@163.com}
                        %  \\
%             \emph{Present address:} of F. Author  %  if needed
        \and
           S.-M. Fei\at
           $^{1}$ School of Mathematical Sciences, Capital Normal University, Beijing, 100048, China\\
           $^{2}$Max-Planck-Institute for Mathematics in the Sciences, 04103 Leipzig, Germany\\
              \email{feishm@mail.cnu.edu.cn}
           }

\date{Received: date / Accepted: date}
% The correct dates will be entered by the editor

\maketitle
\begin{abstract}
We study the constructions of nonlocal orthogonal product states in multipartite systems that cannot be distinguished by local operations and classical communication. We first present two constructions of nonlocal orthogonal product states in tripartite systems $\mathcal{C}^{d}\otimes\mathcal{C}^{d}\otimes\mathcal{C}^{d}~(d\geq3)$ and  $\mathcal{C}^d\otimes \mathcal{C}^{d+1}\otimes \mathcal{C}^{d+2}~(d\geq 3)$. Then for general tripartite quantum system $\mathcal{C}^{n_{1}}\otimes\mathcal{C}^{n_{2}}\otimes\mathcal{C}^{n_{3}}$ $(3\leq n_{1}\leq n_{2}\leq n_{3})$, we obtain $2(n_{2}+n_{3}-1)-n_{1}$ nonlocal orthogonal product states. Finally, we put forward a new construction approach in $\mathcal{C}^{d_{1}}\otimes \mathcal{C}^{d_{2}}\otimes\cdots\otimes \mathcal{C}^{d_{n}}$ $(d_1,d_2,\cdots d_n\geq3,\, n>6)$ multipartite systems. Remarkably, our indistinguishable sets contain less nonlocal product states than the existing ones, which improves the recent results and highlights their related applications in quantum information processing.

% \PACS{PACS code1 \and PACS code2 \and more}
% \subclass{MSC code1 \and MSC code2 \and more}
\end{abstract}

\section{Introduction}
Non-locality is the fundamental feature of quantum mechanics.
The indistinguishability of a set of orthogonal multipartite product states under local operations and classical communication (LOCC)
has attracted much attention as a distinguished phenomenon of non-locality without quantum entanglement \cite{Peres1991,Nielsen1998,Bennett1999,Li2009,Bandy2011,Xu2016,Walgate2000,Horo2003,Chen2004,Feng2009,Yu2012,Nathanson2013,Yang2013,Zhangzc2013,Yu,Zhangzc2014,Zhangzc2015,Wangyl2015,Zhangzc2016,Niset2006,Wangyl2017,Zhangzc2018,Li2019,Xu2020,Xu2021}. Many theoretical achievements have been made in the local discrimination on sets of orthogonal product states as well as orthogonal maximally entangled states.

Local indistinguishability of quantum states is tightly related to the study on the relationship between quantum entanglement and quantum non-locality. Walgate {\it et al}. \cite{Walgate2000} proved that any two orthogonal quantum states, either entangled or separable, can always be distinguished perfectly, showing that quantum entanglement is not a sufficient condition for non-locality. In recent years, many results about non-locality without entanglement in bipartite systems have been obtained. Zhang {\it et al}. \cite{Zhangzc2014,Zhangzc2015} investigated orthogonal product states in $\mathcal{C}^{d}\otimes \mathcal{C}^{d}$ quantum systems and proved that there are sets of $4d-4$ indistinguishable orthogonal product states. Wang {\it et al}. \cite{Wangyl2015} showed that there exist sets of $3~(m+n)-9$ indistinguishable orthogonal product states in general $\mathcal{C}^{m}\otimes \mathcal{C}^{n}$ bipartite systems. Subsequently, the number of the members in a set of locally indistinguishable orthogonal product states in $\mathcal{C}^{m}\otimes \mathcal{C}^{n} (n>m)$ quantum systems is reduced to $2n-1$ by Zhang {\it et al}. \cite{Zhangzc2016}.

Beyond these advances, a certain number of achievements have been made in the local discrimination of multipartite orthogonal product states. In 2006, Nisetet {\it et al}. \cite{Niset2006} presented a class of LOCC indistinguishable orthogonal product bases for ~$(d_j\geq n-1)$, where $n$ is the number of the subsystems and $d_j$ is the dimension of the $j$th subsystem. Zhang {\it et al}. \cite{Zhangzc2015} constructed locally indistinguishable orthogonal products states $\mathcal{C}^d\otimes \mathcal{C}^d\otimes \mathcal{C}^d$ and locally indistinguishable orthogonal product basis quantum
states in multipartite systems. Then, Xu {\it et al}. \cite{Xu2016} gave a nonlocal set of $2n$ multipartite orthogonal product states in ~$\mathcal{C}^{d_{1}}\otimes \mathcal{C}^{d_{2}}\otimes\cdots\otimes \mathcal{C}^{d_{n}}$ $(d_1,d_2,\cdots, d_n\geq3, n\geq3)$. In \cite{Wangyl2017}, Wang {\it et al}. showed that there are ~$2(n_1+n_3)-3$ indistinguishable orthogonal product states in ~$\mathcal{C}^{n_{1}}\otimes\mathcal{C}^{n_{2}}\otimes\mathcal{C}^{n_{3}}$ $(3\leq n_{1}\leq n_{2}\leq n_{3})$. Recently, Zhang {\it et al}. \cite{Zhangzc2017} put forward three general methods to construct nonlocal multipartite orthogonal product states by using the bipartite LOCC indistinguishable states. And Halder {\it et al}. presented a new construction of ~$2n(d-1)$ locally indistinguishable product states in ~$\mathcal{C}^{d}\otimes \mathcal{C}^{d}\otimes\cdots\otimes \mathcal{C}^{d}$ $(d\geq2)$ in \cite{Halder2018}. Shortly before, Jiang {\it et al}. \cite{Xu2020} presented a construction on sets with smaller number of nonlocal orthogonal product states which cannot be distinguished under LOCC.

An interesting problem with both physical and mathematical significance is how fewer number of elements it could be to make up a nonlocal
set of multipartite orthogonal product states for a given Hilbert space. With a smaller nonlocal set of states, one can obtain more new nonlocal sets of quantum states by adding orthogonal product states to the set. In this paper we mainly explore further construction on the LOCC indistinguishable sets with less member of nonlocal multipartite orthogonal product states. Firstly, we set about the question from the tripartite system  $\mathcal{C}^{d}\otimes\mathcal{C}^{d}\otimes\mathcal{C}^{d}~(d\geq3)$ and  $\mathcal{C}^d\otimes \mathcal{C}^{d+1}\otimes \mathcal{C}^{d+2}~(d\geq 3)$. We then come to conclusions for general tripartite quantum system~$\mathcal{C}^{n_{1}}\otimes\mathcal{C}^{n_{2}}\otimes\mathcal{C}^{n_{3}}$ $(3\leq n_{1}\leq n_{2}\leq n_{3})$. We present the smaller set with $2(n_{2}+n_{3}-1)-n_{1}$ nonlocal orthogonal product states, which is better than the previous results. Moreover, we present our direct constructions of nonlocal orthogonal product states in general arbitrary multipartite quantum systems based on the constructions for tripartite systems.

\section{Constructions in Tripartite Quantum Systems}

In this section, we give new nonlocal set of orthogonal product states in tripartite quantum systems. For simplicity, we choose the computational basis $\{|i\rangle\}_{i=0}^{d-1}$ for each $d$-dimensional subsystem $C^d$ and denote $|i_1\pm i_2\pm...\pm i_n\rangle$ the normalized states $\frac{1}{\sqrt{n}}(|i_1\rangle \pm |i_2\rangle\pm...\pm |i_n\rangle)$ without confusion.

\begin{lemma} The following set of $3d-2$ orthogonal product states in tripartite $\mathcal{C}^d\otimes \mathcal{C}^d\otimes \mathcal{C}^d$ $(d\geq3)$ quantum systems, denoted by $A$, $B$ and $C$,
\begin{equation}
\begin{aligned}
|\phi_{i}\rangle&=|0-i\rangle_{A}|0\rangle_{B}|i\rangle_{C},\\
|\phi_{i+(d-1)}\rangle & =|i\rangle_{A}|0-i\rangle_{B}|0\rangle_{C}, \\
|\phi_{i+2(d-1)}\rangle & =|0\rangle_{A}|i\rangle_{B}|0-i\rangle_{C},\\
|\phi_{3(d-1)+1}\rangle & =|0+\cdots+(d-1)\rangle_{A}|0+\cdots+\\& ~~~~(d-1)\rangle_{B}|0+\cdots+(d-1)\rangle_{C},
\end{aligned}
\end{equation}
$i=1,2,\cdots,d-1$, are indistinguishable under LOCC for arbitrary dimension $d$.
\end{lemma}

\begin{proof} As these states are symmetric, we only need to prove the indistinguishability when Alice implements measurements on the subsystem system $A$ first.
Assume that Alice  applies nontrivial and nondisturbing measurement with the positive operator-valued measure (POVM) elements $M_{A}^{\dagger}M_{A}$ in the computational basis,
$$ M_{A}^{\dagger}M_{A}=\left(
 \begin{array}{cccc}
 a_{00} & a_{01} & \cdots & a_{0,d-1}\\
 a_{10} & a_{11} & \cdots & a_{1,d-1} \\
 \vdots & \vdots & \ddots & \vdots  \\
 a_{d-1,0} & a_{d-1,1} & \cdots & a_{d-1,d-1} \\
 \end{array}
 \right).
$$
The post-measurement states $\{M_{A}\otimes I_{B}\otimes I_{C}|\phi_{i}\rangle$, $i=1,\cdots,3d-2\}$ should be mutually orthogonal.
Consider the states $|\phi_{i+(d-1)}\rangle$ and $|\phi_{i+2(d-1)}\rangle$, $ i=1,2,\cdots,d-1$, we get $\langle i|M_{A}^{\dagger}M_{A}|0\rangle \langle0-i|I_{B}|i\rangle \langle0|I_{C}|0-i\rangle=0$, that is, $\langle i|M_{A}^{\dagger}M_{A}$ $|0\rangle=0$. Therefore, $a_{0i}=a_{i0}=0$ for $i=1,2,\cdots,d-1$.

For the states~$|\phi_{i+(d-1)}\rangle$ and~$|\phi_{j+(d-1)}\rangle$, $ 1 \leq i\neq j\leq d-1$, we have~$\langle i|M_{A}^{\dagger}M_{A}$ $|j\rangle \langle0-i|I_{B}|0-j\rangle \langle0|I_{C}|0\rangle=0$, namely, $\langle i|M_{A}^{\dagger}M_{A}|j\rangle=0$. Hence, $a_{ij}=0$ for $1\leq i\neq j\leq d-1$.

For the states~$|\phi_{i}\rangle$ and~$|\phi_{3(d-1)+1}\rangle$, $1\leq i\leq d-1$, we obtain $\langle 0-i|M_{A}^{\dagger}M_{A}|0+\cdots+(d-1)\rangle \langle0|I_{B}|0+\cdots+(d-1)\rangle \langle i|I_{C}|0+\cdots+(d-1)\rangle=0$, i.e., $\langle 0-i|M_{A}^{\dagger}M_{A}|0+\cdots+(d-1)\rangle=0$, which give rise to $a_{00}=a_{ii}$ for $1\leq i\leq d-1$.

In summary all the Alice's POVM elements $M_{A}^{\dagger}M_{A}$ are proportional to the identity matrix. Thus Alice cannot start with a nontrivial measurement, and neither
Bob nor Charlie can start with the measurements on the subsystems $B$ and $C$, respectively, due to the symmetry of the given states. Therefore, these $3d-2$ states cannot be distinguished by LOCC.
\end{proof}

In Ref. \cite{Xu2020} the authors presented an elegant construction of 10 indistinguishable states under LOCC in $C^3\otimes C^3\otimes C^3$. According to our Lemma 1, we obtain 7 nonlocal orthogonal product states, which is less than the one from \cite{Xu2020}.

\begin{lemma} The following $3d+4$ orthogonal product states in
$\mathcal{C}^d\otimes \mathcal{C}^{d+1}\otimes \mathcal{C}^{d+2}$ $(d\geq3)$ cannot be distinguished under LOCC,
\begin{equation}
\begin{aligned}
|\phi_{i}\rangle & =|0-i\rangle_{A}|0\rangle_{B}|i\rangle_{C},\\
|\phi_{i+(d-1)}\rangle & =|i\rangle_{A}|0-i\rangle_{B}|0\rangle_{C}, \\
|\phi_{i+2(d-1)}\rangle & =|0\rangle_{A}|i\rangle_{B}|0-i\rangle_{C},\\
|\phi_{3(d-1)+1}\rangle & =|0-1\rangle_{A}|0\rangle_{B}|d\rangle_{C},\\
|\phi_{3(d-1)+2}\rangle & =|0-1\rangle_{A}|d\rangle_{B}|0\rangle_{C},\\
|\phi_{3(d-1)+3}\rangle & =|(d-2)\rangle_{A}|(d-2)\rangle_{B}|(d-1)-d\rangle_{C},\\
|\phi_{3(d-1)+4}\rangle & =|(d-2)\rangle_{A}|(d-1)-d\rangle_{B}|(d-2)\rangle_{C},\\
|\phi_{3(d-1)+5}\rangle & =|0-1\rangle_{A}|0\rangle_{B}|(d+1)\rangle_{C},\\
|\phi_{3(d-1)+6}\rangle & =|(d-1)\rangle_{A}|(d-1)\rangle_{B}|d-(d+1)\rangle_{C},\\
|\phi_{3(d-1)+7}\rangle & =|0+\cdots+(d-1)\rangle_{A}|0+\cdots+d\rangle_{B} |0+\\ & ~~~~\cdots+(d+1)\rangle_{C},
\end{aligned}
\end{equation}
where $i=1,2,\cdots,d-1.$
\end{lemma}

The proof of the Lemma 2 is given in Appendix A.
By Lemma 2, we have 13 indistinguishable orthogonal product states in $C^3\otimes C^4\otimes C^5$, which is less than the number 16 given by the result from \cite{Xu2020}.

Based on the above recursive construction method used in constructing the non-local sets in Lemma 1 and 2, we have the following conclusions for general tripartite systems.

\begin{theorem} In $\mathcal{C}^{n_{1}}\otimes\mathcal{C}^{n_{2}}\otimes\mathcal{C}^{n_{3}}$ $(3\leq n_{1}\leq n_{2}\leq n_{3})$ tripartite quantum systems, there exist following $2~(n_{2}+n_{3}-1)-n_{1}$ orthogonal product states that cannot be distinguished under LOCC,
\begin{equation}
\begin{aligned}
|\phi_{i}\rangle & =|0-i\rangle_{A}|0\rangle_{B}|i\rangle_{C},\\
|\phi_{i+(n_1-1)}\rangle & =|i\rangle_{A}|0-i\rangle_{B}|0\rangle_{C}, \\
|\phi_{i+2(n_{1}-1)}\rangle & =|0\rangle_{A}|i\rangle_{B}|0-i\rangle_{C},\\
|\phi_{j+2(n_{1}-1)}\rangle & =|0-1\rangle_{A}|0\rangle_{B}|j\rangle_{C},\\
|\phi_{j+2(n_{1}-1)+(n_{2}-n_{1})}\rangle & =|0-1\rangle_{A}|j\rangle_{B}|0\rangle_{C},\\
|\phi_{j+2(n_{1}-1)+2(n_{2}-n_{1})}\rangle & =|0\rangle_{A}|m\rangle_{B}|(j-1)-j\rangle_{C},\\
|\phi_{j+2(n_{1}-1)+3(n_{2}-n_{1})}\rangle & =|0\rangle_{A}|(j-1)-j\rangle_{B}|m\rangle_{C},\\
|\phi_{k+2(n_{1}-1)+3(n_{2}-n_{1})}\rangle & =|0-1\rangle_{A}|0\rangle_{B}|k\rangle_{C},\\
|\phi_{k-n_{1}+2n_{2}+n_{3}-2}\rangle & =|0\rangle_{A}|m\rangle_{B}|(k-1)-k\rangle_{C},\\
|\phi_{2(n_{2}+n_{3}-1)-n_{1}}\rangle & =|0+\cdots+(n_{1}-1)\rangle_{A}
|0+\cdots+ \\ &~~~~(n_{2}-1)\rangle_{B}|0+\cdots+(n_{3}-1)\rangle_{C},
\end{aligned}
\end{equation}
where $i=1,\cdots,n_1-1,$ $j=n_{1},\cdots,n_{2}-1,$ and $k=n_{2},\cdots,n_{3}-1$. When $j-1$ is even, $m=1$; when $j-1$ is odd, $m=2$. And when $k-1$ is even, $m=1$; when $k-1$ is odd, $m=2$.
\end{theorem}

\begin{proof} Without loss of generality, we suppose that Alice starts with the nontrivial and nondisturbing measurements first, with the POVM elements $M_{A}^{\dagger}M_{A}$ in the computational basis as follows:\\

$M_{A}^{\dagger}M_{A}=\left(
\begin{array}{cccc}
a_{00} & a_{01} & \cdots & a_{0,n_{1}-1}\\
a_{10} & a_{11} & \cdots & a_{1,n_{1}-1} \\
\vdots & \vdots & \ddots & \vdots  \\
a_{n_{1}-1,0} & a_{n_{1}-1,1} & \cdots & a_{n_{1}-1,n_{1}-1} \\
\end{array}
\right).
$

The post-measurement states $\{M_{A}\otimes I_{B}\otimes I_{C}|\phi_{i}\rangle,$ $i=1,\cdots,2(n_{2}+n_{3}-1)-n_{1}\}$ should be mutually orthogonal. Considering the states~$|\phi_{i+(n_{1}-1)}\rangle$ and~$|\phi_{i+2(n_{1}-1)}\rangle$, $i=1,2,\cdots,n_{1}-1$, we have $\langle i|M_{A}^{\dagger}M_{A}|0\rangle \langle0-i|I_{B}|i\rangle \langle0|I_{C}|0-i\rangle=0$, that is,~$\langle i|M_{A}^{\dagger}M_{A}|0\rangle=0$. We have $a_{0i}=a_{i0}=0$ for $i=1,2,\cdots,n_{1}-1$.

For the states~$|\phi_{i+(n_{1}-1)}\rangle$ and~$|\phi_{j+(n_{1}-1)}\rangle$,~$ i\neq j,$ $i,j=1,2,\cdots,n_{1}-1$, we have $\langle i|M_{A}^{\dagger}M_{A}|j\rangle \langle0-i|I_{B}|0-j\rangle \langle0|I_{C}|0\rangle=0$, i.e., $\langle i|M_{A}^{\dagger}M_{A}|j\rangle=0$. Hence, we get $a_{ij}=0$ for $1\leq i\neq j \leq n_{1}-1$.

For the states~$|\phi_{i}\rangle$ and $|\phi_{2(n_{2}+n_{3}-1)-n_{1}}\rangle$, $ i=1,2,\cdots,n_{1}-1$, we obtain $\langle 0-i|M_{A}^{\dagger}M_{A}|0+\cdots+(n_{1}-1)\rangle \langle0|I_{B}|0+\cdots+(n_{2}-1)\rangle \langle i|I_{C}|0+\cdots+(n_{3}-1)\rangle=0$, namely,~$\langle 0-i|M_{A}^{\dagger}M_{A}|0+\cdots+(n_{1}-1)\rangle=0$, which give rise to $a_{00}=a_{ii}$ for all $i=1,2,\cdots,n_{1}-1$.

From the above analysis it turns out that all the Alice's POVM elements $M_{A}^{\dagger}M_{A}$ are proportional to the identity matrix, which shows that Alice cannot perform a nontrivial measurement on her subsystem. Similar results apply to Bob and Charlie. Therefore, the~$2~(n_{2}+n_{3}-1)-n_{1}$ states cannot be locally distinguished by LOCC.

The impossibility of nontrivial measurement for Bob and Charlie can be proved in a similar way, see Appendix B.
\end{proof}

Let us compare our conclusions from Theorem 1 with the existing ones. In Ref. \cite{Wangyl2017}, Wang {\it et al}. constructed ~$2(n_1+n_3)-3$ indistinguishable orthogonal product state in ~$\mathcal{C}^{n_{1}}\otimes\mathcal{C}^{n_{2}}\otimes\mathcal{C}^{n_{3}}~(3\leq n_{1}\leq n_{2}\leq n_{3})$. It is easily verified that if $n_{2}<\frac{3}{2}n_{1}-\frac{1}{2}$, then
$(2(n_{2}+n_{3}-1)-n_{1})<(2(n_{1}+n_{3})-3)$, that is, the nonlocal orthogonal product states sets we constructed have less elements than the previous results. In particular, when the first two subsystems have same dimensions ($n_1=n_2$), our construction has $n_1-1$ states less than the result from \cite{Wangyl2017}. Moreover, when $n_2=3$, our construction has always less states than the one in \cite{Wangyl2017} for arbitrary dimensions $n_1$ and $n_3$.

\section{Constructions in Multipartite Quantum Systems}

Based on the results for tripartite systems, we put forward our new approach in constructing nonlocal orthogonal product states for multipartite quantum systems. For clarity, we deal with the problem for general $n$-partite systems in~$\mathcal{C}^{d_{1}}\otimes \mathcal{C}^{d_{2}}\otimes\cdots\otimes \mathcal{C}^{d_{n}}$ by three detailed cases: $n=0~(mod~3)$, $n=1~(mod~3)$ and $n=2~(mod~3)$.

\begin{theorem} Suppose that $\{|\varphi\rangle_{it}=|x\rangle_{it}|y\rangle_{it}|z\rangle_{it}$, $i=1,2,\cdots, l_{t}\}$ is a set of locally indistinguishable orthogonal product states in $\mathcal{C}^{d_{t}}\otimes \mathcal{C}^{d_{t+1}}\otimes \mathcal{C}^{d_{t+2}}$ $(t=1,4,\cdots,n-2)$. For quantum systems in $\mathcal{C}^{d_{1}}\otimes \mathcal{C}^{d_{2}}\otimes\cdots\otimes \mathcal{C}^{d_{n}}$, with $d_1,d_2,\cdots, d_n\geq3$, $n>6$ and $n\equiv0~(mod~3)$, the following $l_{1}+l_{4}+\cdots+l_{n-2}$ orthogonal product states are indistinguishable under LOCC,
\begin{small}
\begin{equation}
\begin{aligned}
|\phi_{i}\rangle & =|x\rangle_{i1}|y\rangle_{i1}|z\rangle_{i1}|b\rangle_{4}|b\rangle_{5}|b\rangle_{6}
\cdots|b\rangle_{n-3}|a\rangle_{n-2}|a\rangle_{n-1}|a\rangle_{n},\\
&~~~~~ i=1,2,\cdots, l_{1};\\
|\phi_{i+l_{1}}\rangle & =|a\rangle_{1}|a\rangle_{2}|a\rangle_{3}|x\rangle_{i4}|y\rangle_{i4}
|z\rangle_{i4}|b\rangle_{7}\cdots|b\rangle_{n-1}|b\rangle_{n},\\
&~~~~~ i=1,2,\cdots, l_{4};\\
& ~~~~~~~~\cdots~ \cdots ~\cdots~~~~~~~~~~~~\\
|\phi_{i+l_{1}+\cdots+l_{m-3}}\rangle  & =|b\rangle_{1}\cdots|b\rangle_{m-4}
|a\rangle_{m-3}|a\rangle_{m-2}|a\rangle_{m-1}
|x\rangle_{im}|y\rangle_{im}|z\rangle_{im}
\\
&~~~~|b\rangle_{m+3}\cdots|b\rangle_{n},~i=1,2,\cdots, l_{m};\\
& ~~~~~~~~ \cdots ~\cdots~ \cdots~~~~~~~~~~~~~\\
|\phi_{i+l_{1}+\cdots+l_{n-5}}\rangle & =|b\rangle_{1}\cdots|b\rangle_{n-6}|a\rangle_{n-5}|a\rangle_{n-4}
|a\rangle_{n-3}|x\rangle_{i(n-2)}|y\rangle_{i(n-2)}|z\rangle_{i(n-2)},\\ &~~~~i=1,2,\cdots, l_{n-2},
\end{aligned}
\end{equation}
\end{small}
where~$m=7,10,\cdots,n-5$, and~$|a\rangle_{t}|a\rangle_{t+1}|a\rangle_{t+2}$ is orthogonal to~$|b\rangle_{t}|b\rangle_{t+1}|b\rangle_{t+2}$ $(t=1,4,\cdots,n-2)$.
\end{theorem}

\begin{proof} To distinguish these states, the first participant must start with a nontrivial measurement. Since  $\{|\varphi\rangle_{it}=|x\rangle_{it}|y\rangle_{it}|z\rangle_{it},$ $ i=1,2,\cdots, l_{t}\}$ is a set of locally indistinguishable orthogonal product states in~$\mathcal{C}^{d_{t}}\otimes \mathcal{C}^{d_{t+1}}\otimes \mathcal{C}^{d_{t+2}}(t=1,4,\cdots,n-2)$, for $|\phi_{i}\rangle$, $i=1,2,\cdots l_{1}$, we can draw the conclusion that any first three participants cannot start with a nontrivial measurement. Similarly, any followed three participants can neither start with a nontrivial measurement for $|\phi_{i+l_{1}}\rangle$, $i=1,2,\cdots l_{4}$. Consequently, for states $|\phi_{i+l_{1}+\cdots+l_{n-5}}\rangle,$ $i=1,2,\cdots l_{n-2}$, the three participants cannot start with a nontrivial measurement either. To sum up, none of the participants can make nontrivial measurements. Therefore, the $l_{1}+l_{4}+\cdots+l_{n-2}$ quantum states in~$\mathcal{C}^{d_{1}}\otimes \mathcal{C}^{d_{2}}\otimes\cdots\otimes \mathcal{C}^{d_{n}}$ $(d_1,d_2,\cdots, d_n\geq3,\, n>6,\, n\equiv0~(mod~3))$ quantum systems cannot be distinguished by LOCC.
\end{proof}

Similarly, we have the following theorems for the rest two cases, see proofs in Appendix C and Appendix D, respectively.

\begin{theorem}Suppose that~$\{|\varphi\rangle_{it}=|x\rangle_{it}|y\rangle_{it}|z\rangle_{it},$ $ i=1,2,\cdots, l_{t}\}$, is a set of locally indistinguishable orthogonal product states in~$\mathcal{C}^{d_{t}}\otimes \mathcal{C}^{d_{t+1}}\otimes \mathcal{C}^{d_{t+2}}$ $(t=1,4,\cdots,n-3)$. And we select~$\{|\varphi\rangle_{in}=|x\rangle_{in}|y\rangle_{in}|z\rangle_{in},$ $ i=1,2,\cdots, l_{n}\}$ is a set of orthogonal product states in~$\mathcal{C}^{d_{1}}\otimes \mathcal{C}^{d_{2}}\otimes \mathcal{C}^{d_{n}}$. The following $l_{1}+l_{4}+\cdots+l_{n}$ orthogonal product states in~$\mathcal{C}^{d_{1}}\otimes \mathcal{C}^{d_{2}}\otimes\cdots\otimes \mathcal{C}^{d_{n}}$ $(d_1,d_2,\cdots,d_ n\geq3,\, n>6,\, n\equiv1~(mod~3))$ are indistinguishable under LOCC,
\begin{small}
\begin{equation}
\begin{aligned}
|\phi_{i}\rangle & =|x\rangle_{i1}|y\rangle_{i1}|z\rangle_{i1}|b\rangle_{4}|b\rangle_{5}
|b\rangle_{6}\cdots|b\rangle_{n-1}|a\rangle_{n}, ~i=1,2,\cdots l_{1};\\
|\phi_{i+l_{1}}\rangle & =|a\rangle_{1}|a\rangle_{2}|a\rangle_{3}|x\rangle_{i4}|y\rangle_{i4}|z\rangle_{i4}
|b\rangle_{7}\cdots|b\rangle_{n-1}|b\rangle_{n},~i=1,2,\cdots l_{4};\\
&~~~~~~~~~\cdots ~\cdots~ \cdots~~~~~~~~~~~~~\\
|\phi_{i+l_{1}+\cdots+l_{m-3}}\rangle & =|b\rangle_{1}\cdots|b\rangle_{m-4}|a\rangle_{m-3}
|a\rangle_{m-2}|a\rangle_{m-1}
|x\rangle_{im}|y\rangle_{im}|z\rangle_{im}\\
&~~~~~|b\rangle_{m+3}\cdots|b\rangle_{n},~i=1,2,\cdots l_{m};\\
&~~~~~~~~~\cdots ~\cdots~ \cdots~~~~~~~~~~~~~\\
|\phi_{i+l_{1}+\cdots+l_{n-3}}\rangle & =|x\rangle_{in}|y\rangle_{in}|b\rangle_{3}\cdots|b\rangle_{n-4}
|a\rangle_{n-3}|a\rangle_{n-2}|a\rangle_{n-1}|z\rangle_{in}, ~~i=1,2,\cdots l_{n},
\end{aligned}
\end{equation}
\end{small}
where $m=7,10,\cdots,n-3$, $|a\rangle_{t}|a\rangle_{t+1}|a\rangle_{t+2}$ is orthogonal to~$|b\rangle_{t}|b\rangle_{t+1}|b\rangle_{t+2}$,~$(t=1,4,\cdots,n-3)$, $|a\rangle_{n}$ and $|b\rangle_{n}$ are orthogonal.
\end{theorem}

\begin{theorem} Suppose that~$\{|\varphi\rangle_{it}=|x\rangle_{it}|y\rangle_{it}|z\rangle_{it}$, $i=1,2,\cdots, l_{t}\}$, is a set of locally indistinguishable orthogonal product states in~$\mathcal{C}^{d_{t}}\otimes \mathcal{C}^{d_{t+1}}\otimes \mathcal{C}^{d_{t+2}}$ $(t=1,4,\cdots,n-4)$. And we select~$\{|\varphi\rangle_{i(n-1)}=|x\rangle_{i(n-1)}|y\rangle_{i(n-1)}|z\rangle_{i(n-1)},$ $ i=1,2,\cdots, l_{n-1}\}$ is a set of orthogonal product states in~$\mathcal{C}^{d_{1}}\otimes \mathcal{C}^{d_{(n-1)}}\otimes \mathcal{C}^{d_{n}}$. The following $l_{1}+l_{4}+\cdots+l_{n-1}$ orthogonal product states in $\mathcal{C}^{d_{1}}\otimes \mathcal{C}^{d_{2}}\otimes\cdots\otimes \mathcal{C}^{d_{n}}$ $(d_1,d_2,\cdots,d_n\geq3,\, n>6,\, n\equiv2~(mod~3))$ are indistinguishable under LOCC,
\begin{small}
\begin{equation}
\begin{aligned}
|\phi_{i}\rangle & =|x\rangle_{i1}|y\rangle_{i1}|z\rangle_{i1}|b\rangle_{4}\cdots|b\rangle_{n-2}
|a\rangle_{n-1}|a\rangle_{n}, ~i=1,2,\cdots l_{1};\\
|\phi_{i+l_{1}}\rangle & =|a\rangle_{1}|a\rangle_{2}|a\rangle_{3}|x\rangle_{i4}|y\rangle_{i4}|z\rangle_{i4}
|b\rangle_{7}\cdots|b\rangle_{n-1}|b\rangle_{n}, ~i=1,2,\cdots l_{4};\\
&~~~~~~~~~\cdots ~\cdots~ \cdots~~~~~~~~~~~~~\\
|\phi_{i+l_{1}+\cdots+l_{m-3}}\rangle & =|b\rangle_{1}\cdots|b\rangle_{m-4}
|a\rangle_{m-3}|a\rangle_{m-2}|a\rangle_{m-1}|x\rangle_{im}|y\rangle_{im}|z\rangle_{im}
\\&~~~~ |b\rangle_{m+3}\cdots|b\rangle_{n},~i=1,2,\cdots l_{m};\\
&~~~~~~~~~\cdots ~\cdots~ \cdots~~~~~~~~~~~~~\\
|\phi_{i+l_{1}+\cdots+l_{n-4}}\rangle & =|x\rangle_{i(n-1)}|b\rangle_{2}\cdots|b\rangle_{n-5}
|a\rangle_{n-4}|a\rangle_{n-3}|a\rangle_{n-2}|y\rangle_{i(n-1)}|z\rangle_{i(n-1)},
\\&~~~~ i=1,2,\cdots l_{n-1},
\end{aligned}
\end{equation}
\end{small}
where $m=7,10,\cdots,n-4$ and $|a\rangle_{t}|a\rangle_{t+1}|a\rangle_{t+2}$ is orthogonal to~$|b\rangle_{t}|b\rangle_{t+1}|b\rangle_{t+2}$ $(t=1,4,\cdots,n-4)$, $|a\rangle_{n-1}|a\rangle_{n}$ and $|b\rangle_{n-1}|b\rangle_{n}$ are orthogonal..
\end{theorem}

As applications we consider a simple example in \cite{Xu2020}. Take $n=9$ and $d=3$ , according to \cite{Xu2020}, one has the sets of nonlocal orthogonal product states, containing 36 and 28 states, respectively. Nevertheless, from Theorem 2 we only need 21 states to construct the set of nonlocal orthogonal product states.

Furthermore, in Ref. \cite{Xu2020} the construction of nonlocal orthogonal product states with $n(2d-3)+1$ members is presented  in $\otimes_{j=1}^n C^d$. According to Theorem 2, we present $\frac{n}{3}(3d-2)$ orthogonal product states in $\otimes_{j=1}^n C^d$ for $n\equiv0~(mod~3)(n>6)$. Similarly, the nonlocal orthogonal product state sets with $\frac{n+2}{3}(3d-2)$ and $\frac{n+1}{3}(3d-2)$ members are constructed by Theorem 3 and Theorem 4 in $\otimes_{j=1}^n C^d$  for $n\equiv1~(mod~3)$ and $ n\equiv2~(mod~3)$ $(d\geq3,~n>6)$, respectively. It is easily proved that our nonlocal sets always have less numbers than the given results.

\section{Conclusions and discussions}

Local discrimination of quantum states has
attracted much attention during the last twenty years.
The local distinguishability of quantum states can be applied to design quantum protocols such as quantum cryptography \cite{Guo2001,Rahaman2015,Wangjt2017,Yang2015}.
The construction of sets of locally indistinguishable multipartite orthogonal product
states with less members is more difficult than bipartite ones.
We have presented improved constructions of locally indistinguishable orthogonal product states in multipartite systems in a simpler and more effective way.
We have constructed $3d-2$ nonlocal orthogonal product states in $\mathcal{C}^d\otimes \mathcal{C}^d\otimes \mathcal{C}^d$ $(d\geq3)$ and $3d+4$  nonlocal orthogonal product states in $\mathcal{C}^d\otimes \mathcal{C}^{d+1}\otimes \mathcal{C}^{d+2}$ $(d\geq3)$. We have found further $2~(n_{2}+n_{3}-1)-n_{1}$ nonlocal orthogonal product states in the tripartite quantum system $\mathcal{C}^{n_{1}}\otimes\mathcal{C}^{n_{2}}\otimes\mathcal{C}^{n_{3}}$ $(3\leq n_{1}\leq n_{2}\leq n_{3})$ with arbitrary dimensions of individual subsystems. Based on the tripartite constructions, we have put forward our recursive construction for arbitrary multipartite systems $\mathcal{C}^{d_{1}}\otimes \mathcal{C}^{d_{2}}\otimes\cdots\otimes \mathcal{C}^{d_{n}}$ $(d_1,d_2,\cdots, d_n\geq3,\, n>6)$. Above all, the LOCC indistinguishable sets we constructed contain less members of states than the existing ones, which optimize further the recent results and would highlight the related researches in quantum information processing.

\begin{acknowledgements}
This work is supported by NSFC (Grant No. 11871019, 12075159), Natural Science Foundation of Hebei Province (F2021205001), Key Project of Beijing Municipal Commission of Education (KZ201810028042), Beijing Natural Science Foundation (Z190005), Academy for Multidisciplinary Studies, Capital Normal University, and Shenzhen Institute for Quantum Science and Engineering, Southern University of Science and Technology, Shenzhen 518055, China (No. SIQSE202001).
\end{acknowledgements}

% BibTeX users please use one of
%\bibliographystyle{spbasic}      % basic style, author-year citations
%\bibliographystyle{spmpsci}      % mathematics and physical sciences
%\bibliographystyle{spphys}       % APS-like style for physics
%\bibliography{}   % name your BibTeX data base

% Non-BibTeX users please use

\bigskip
\section*{APPENDIX}
\setcounter{equation}{0} \renewcommand%
\theequation{A\arabic{equation}}

{\bf A. Proof of Lemma 2}

\begin{proof}
~$(1)$ When Alice starts with the nontrivial and nondisturbing measurement~$M_{A}^{\dagger}M_{A}$, we write the POVM elements $M_{A}^{\dagger}M_{A}$ in the~$\{|0\rangle,|1\rangle,\cdots,|d-1\rangle\}_{A}$ basis:
$$ M_{A}^{\dagger}M_{A}=\left(
                                   \begin{array}{cccc}
                                     a_{00} & a_{01} & \cdots & a_{0,d-1}\\
                                     a_{10} & a_{11} & \cdots & a_{1,d-1} \\
                                     \vdots & \vdots & \ddots & \vdots  \\
                                     a_{d-1,0} & a_{d-1,1} & \cdots & a_{d-1,d-1} \\
                                   \end{array}
                                 \right),
 $$

The postmeasurement states~$\{M_{A}\otimes I_{B}\otimes I_{C}|\phi_{i}\rangle, i=1,\cdots,3d+4\}$ should be mutually orthogonal. Consider the states $|\phi_{i+(d-1)}\rangle$ and ~$|\phi_{i+2(d-1)}\rangle$,~$i=1,2,\cdots,d-1$, we have $\langle i|M_{A}^{\dagger}M_{A}|0\rangle \langle0-i|I_{B}|i\rangle \langle0|I_{C}|0-i\rangle=0$, i.e., ~$\langle i|M_{A}^{\dagger}M_{A}|0\rangle$ $=0$, so~$a_{0i}=a_{i0}=0, i=1,2,\cdots,d-1$.

For the states~$|\phi_{i+(d-1)}\rangle$ and $|\phi_{j+(d-1)}\rangle$, $ i\neq j,$ $i,j=1,2,\cdots,d-1$, we have $\langle i|M_{A}^{\dagger}M_{A}|j\rangle \langle0-i|I_{B}|0-j\rangle \langle0|I_{C}|0\rangle=0$, that is, $\langle i|M_{A}^{\dagger}M_{A}|j\rangle=0$, so $a_{ij}=0$, $i\neq j,$ $i,j=1,2,\cdots,d-1$.

For the states~$|\phi_{i}\rangle$ and~$|\phi_{3(d-1)+7}\rangle$, $ i=1,2,\cdots,d-1$, we have $\langle 0-i|M_{A}^{\dagger}M_{A}|0+\cdots+(d-1)\rangle \langle0|I_{B}|0+\cdots+d\rangle \langle i|I_{C}|0+\cdots+(d+1)\rangle=0$, that is,  $\langle 0-i|M_{A}^{\dagger}M_{A}|0+\cdots+(d-1)\rangle=0$, so $a_{00}=a_{ii}=0$, $i=1,2,\cdots,d-1$.

Therefore, all of Alice's POVM elements $M_{A}^{\dagger}M_{A}$ are proportional to the identity operator and Alice cannot start with a nontrivial  measurement.

~$(2)$ As for Bob, we write the POVM elements $M_{B}^{\dagger}M_{B}$ in the~$\{|0\rangle,|1\rangle,\cdots,$ $|d\rangle\}_{B}$ basis:
$$ M_{B}^{\dagger}M_{B}=\left(
\begin{array}{ccccc}
b_{00} & b_{01} & \cdots & b_{0,d-1} & b_{0,d}\\
b_{10} & b_{11} & \cdots & b_{1,d-1} & b_{1,d}\\
\vdots & \vdots & \ddots & \vdots   & \vdots\\
b_{d-1,0} & b_{d-1,1} & \cdots & b_{d-1,d-1} & b_{d-1,d}\\
b_{d0} & b_{d1} & \cdots & b_{d,d-1} & b_{dd}\\
\end{array}
\right),
$$

The postmeasurement states~$\{I_{A}\otimes M_{B}\otimes I_{C}|\phi_{i}\rangle, i=1,\cdots,3d+4\}$ should be mutually orthogonal. Consider the states $|\phi_{i}\rangle$ and ~$|\phi_{i+2(d-1)}\rangle, i=1,\cdots,d-1$, we have $\langle0-i|I_{A}|0\rangle \langle0|M_{B}^{\dagger}M_{B}|i\rangle \langle i|I_{C}|0-i\rangle=0$, i.e., $\langle 0|M_{B}^{\dagger}M_{B}|i\rangle=0$, so~$b_{0i}=b_{i0}=0$, $i=1,\cdots,d-1$.

For the states~$|\phi_{i+2(d-1)}\rangle$ and~$|\phi_{j+2(d-1)}\rangle, i\neq j,~~i,j=1,2,\cdots,d-1$, we have $\langle0|I_{A}|0\rangle \langle i|M_{B}^{\dagger}M_{B}|j\rangle \langle0-i|I_{C}|0-j\rangle=0$, that is, $\langle i|M_{B}^{\dagger}M_{B}|j\rangle=0$, so $b_{ij}=b_{ji}=0$, $i\neq j,$ $i,j=1,2,\cdots,d-1$.

For the states $|\phi_{i+2(d-1)}\rangle$ and~$|\phi_{3(d-1)+2}\rangle, i=1,2,\cdots,d-1$, we have $\langle0|I_{A}|0-1\rangle \langle i|M_{B}^{\dagger}M_{B}|d\rangle \langle0-i|I_{C}|0\rangle=0$, that is, $\langle i|M_{B}^{\dagger}M_{B}|d\rangle=0$, so $b_{id}=b_{di}=0$, $i=1,2,\cdots,d-1$.

For the states $|\phi_{1+(d-1)}\rangle$ and~$|\phi_{3(d-1)+2}\rangle, i=1,2,\cdots,d-1$,  we have $\langle1|I_{A}|0-1\rangle \langle 0-1|M_{B}^{\dagger}M_{B}|d\rangle \langle0|I_{C}|0\rangle=0$, that is, $\langle 0-1|M_{B}^{\dagger}M_{B}|d\rangle=0$, so $b_{0d}=b_{d0}=b_{1d}=0$.

For the states $|\phi_{i+(d-1)}\rangle$ and $|\phi_{3(d-1)+7}\rangle, i=1,\cdots,d-1$, we have $\langle i|I_{A}|0+\cdots+(d-1)\rangle \langle0-i|M_{B}^{\dagger}M_{B}|0+\cdots+d\rangle \langle 0|I_{C}|0+\cdots+(d+1)\rangle=0$, that is, $\langle 0-i|M_{B}^{\dagger}M_{B}|0+\cdots+d\rangle=0$, so~$b_{00}=b_{ii}, i=1,\cdots,d-1$. And for the states~$|\phi_{3(d-1)+4}\rangle$ and~$|\phi_{3(d-1)+7}\rangle$, we also have $\langle (d-2)|I_{A}|0+\cdots+(d-1)\rangle \langle(d-1)-d|M_{B}^{\dagger}M_{B}|0+\cdots+d\rangle \langle (d-2)|I_{C}|0+\cdots+(d+1)\rangle=0$, i.e.,~$\langle (d-1)-d|M_{B}^{\dagger}M_{B}|0+\cdots+d\rangle=0$, therefore~$b_{d-1,d-1}=b_{dd}$. Thus we assert that the diagonal elements of~$M_{B}^{\dagger}M_{B}$ are equal.

Therefore, all of Bob's POVM elements $M_{B}^{\dagger}M_{B}$ are proportional to the identity and Bob cannot start with a nontrivial  measurement.

~$(3)$ As for Charlie, we write the POVM elements $M_{C}^{\dagger}M_{C}$ in the $\{|0\rangle,|1\rangle,\cdots,$ $|d+1\rangle\}_{C}$ basis:
$$ M_{C}^{\dagger}M_{C}=\left(
\begin{array}{ccccc}
c_{00} &  \cdots & c_{0,d-1} & c_{0d} & c_{0,d+1}\\
c_{10} &  \cdots & c_{1,d-1} & c_{1d} & c_{1,d+1}\\
\vdots &  \ddots & \vdots   & \vdots & \vdots \\
c_{d-1,0}& \cdots&c_{d-1,d-1}&c_{d-1,d}&c_{d-1,d+1}\\
c_{d 0} &\cdots &c_{d,d-1}&c_{d d}& c_{d,d+1}\\
c_{d+1,0}&\cdots&c_{d+1,d-1}&c_{d+1,d}&c_{d+1,d+1}\\
\end{array}
\right),
 $$

The postmeasurement states~$\{I_{A}\otimes I_{B}\otimes M_{C}|\phi_{i}\rangle, i=1,\cdots,3d+4\}$ should be mutually orthogonal. Consider the states $|\phi_{i}\rangle$ and ~$|\phi_{i+(d-1)}\rangle, i=1,\cdots,d-1$, we have $\langle0-i|I_{A}|i\rangle \langle0|I_{B}|0-i\rangle \langle i|M_{C}^{\dagger}M_{C}|0\rangle=0$, i.e., $\langle i|M_{C}^{\dagger}M_{C}|0\rangle=0$, so $c_{0i}=c_{i0}=0$, $i=1,\cdots,d-1$.

For the states $|\phi_{i}\rangle$ and~$|\phi_{j}\rangle$, $i\neq j$, $i,j=1,2,\cdots,d-1$, we have $\langle0-i|I_{A}|0-j\rangle \langle0|I_{B}|0\rangle \langle i|M_{C}^{\dagger}M_{C}|j\rangle=0$, that is, $\langle i|M_{C}^{\dagger}M_{C}|j\rangle=0$, so $c_{ij}=c_{ji}=0$, $i\neq j$, $i,j=1,2,\cdots,d-1$.

For the states~$|\phi_{1+(d-1)}\rangle$ and~$|\phi_{3(d-1)+1}\rangle$, we have $\langle1|I_{A}|0-1\rangle \langle0-1|I_{B}|0\rangle$ $\langle 0|M_{C}^{\dagger}M_{C}|d\rangle=0,$ i.e., $\langle 0|M_{C}^{\dagger}M_{C}|d\rangle=0$, so $c_{0d}=c_{d0}=0$.

For the states $|\phi_{i}\rangle$ and $|\phi_{3(d-1)+1}\rangle$, $i=1,2,\cdots,d-1$, we have $\langle0-i|I_{A}|0-1\rangle \langle0|I_{B}|0\rangle \langle i|M_{C}^{\dagger}M_{C}|d\rangle=0$, that is, $\langle i|M_{C}^{\dagger}M_{C}|d\rangle=0$, so $c_{id}=c_{di}=0$, $i=1,2,\cdots,d-1$.

For the states $|\phi_{1+(d-1)}\rangle$ and~$|\phi_{3(d-1)+5}\rangle$, we have $\langle1|I_{A}|0-1\rangle \langle0-1|I_{B}|0\rangle$ $\langle 0|M_{C}^{\dagger}M_{C}|(d+1)\rangle=0$, that is, $\langle 0|M_{C}^{\dagger}M_{C}|(d+1)\rangle=0$, so $c_{0,d+1}=c_{d+1,0}=0$.

For the states $|\phi_{i}\rangle$ and~$|\phi_{3(d-1)+5}\rangle$, $i=1,2,\cdots,d-1$, we have $\langle0-i|I_{A}|0-1\rangle \langle0|I_{B}|0\rangle \langle i|M_{C}^{\dagger}M_{C}|(d+1)\rangle=0$, i.e., $\langle i|M_{C}^{\dagger}M_{C}|(d+1)\rangle=0$, so $c_{i,d+1}=c_{d+1,i}=0$, $i=1,2,\cdots,d-1$.

For the states~$|\phi_{3(d-1)+1}\rangle$ and $|\phi_{3(d-1)+5}\rangle$, $i=1,2,\cdots,d-1$, we have $\langle0-1|I_{A}|0-1\rangle \langle0|I_{B}|0\rangle \langle d|M_{C}^{\dagger}M_{C}|d+1\rangle=0$, that is, $\langle d|M_{C}^{\dagger}M_{C}|d+1\rangle=0$, so $c_{d,d+1}=c_{d+1,d}=0$.

For the states $|\phi_{i+2(d-1)}\rangle$ and~$|\phi_{3(d-1)+7}\rangle$, $i=1,\cdots,d-1$, we have $\langle0|I_{A}|0+\cdots+(d-1)\rangle \langle i|I_{B}|0+\cdots+d\rangle \langle 0-i|M_{C}^{\dagger}M_{C}|0+\cdots+(d+1)\rangle=0$, that is, $\langle 0-i|M_{C}^{\dagger}M_{C}|0+\cdots+(d+1)\rangle=0$, so $c_{00}=c_{ii}$, $i=1,\cdots,d-1$.

For the states $|\phi_{3(d-1)+3}\rangle$ and $|\phi_{3(d-1)+7}\rangle$, we have $\langle(d-2)|I_{A}|0+\cdots+(d-1)\rangle \langle (d-2)|I_{B}|0+\cdots+d\rangle \langle (d-1)-d|M_{C}^{\dagger}M_{C}|0+\cdots+(d+1)\rangle=0$, i.e., $\langle (d-1)-d|M_{C}^{\dagger}M_{C}|0+\cdots+(d+1)\rangle=0$, so $c_{d-1,d-1}=c_{dd}$. And for the states $|\phi_{3(d-1)+6}\rangle$ and~$|\phi_{3(d-1)+7}\rangle$, we have $\langle(d-1)|I_{A}|0+\cdots+(d-1)\rangle \langle (d-1)|I_{B}|0+\cdots+d\rangle \langle d-(d+1)|M_{C}^{\dagger}M_{C}|0+\cdots+(d+1)\rangle=0$, that is, $\langle d-(d+1)|M_{C}^{\dagger}M_{C}|0+\cdots+(d+1)\rangle=0$, so $c_{dd}=c_{d+1,d+1}$. Thus the diagonal elements of $M_{C}^{\dagger}M_{C}$ are equal.

Therefore, all of Charlie's POVM elements $M_{C}^{\dagger}M_{C}$ are proportional to the identity matrices and Charlie cannot start with a nontrivial measurement.

To sum up, all of the three participators Alice, Bob and Charlie cannot start with a nontrivial measurement. So, the~$3d+4$ states cannot be perfectly distinguished by LOCC.
\end{proof}

\noindent {\bf B. The impossibility of nontrivial measurement for Bob and Charlie}

As for Bob, we write the POVM elements $M_{B}^{\dagger}M_{B}$ in the basis $\{|0\rangle,|1\rangle,|2\rangle,\cdots,$ $|n_{2}-1\rangle\}_{B}$:

$\left(
\begin{array}{cccccc}
b_{00}&\cdots&b_{0,n_{1}-1}&b_{0n_{1}}&\cdots&b_{0,n_{2}-1}\\
b_{10}&\cdots&b_{1,n_{1}-1}&b_{1n_{1}}&\cdots& b_{1,n_{2}-1}\\
\vdots&    &\vdots  &\vdots&       &\vdots\\
b_{n_{1}-1,0}&\cdots&b_{n_{1}-1,n_{1}-1}&b_{n_{1}-1,n_{1}}&\cdots&b_{n_{1}-1,n_{2}-1}\\
b_{n_{1}0} &\cdots&b_{n_{1},n_{1}-1}&b_{n_{1}n_{1}}&\cdots&b_{n_{1},n_{2}-1}\\
\vdots&   &\vdots  &\vdots&       &\vdots\\
b_{n_{2}-1,0}&\cdots&b_{n_{2}-1,n_{1}-1}&b_{n_{2}-1,n_{1}}&\cdots&b_{n_{2}-1,n_{2}-1}\\
\end{array}
\right),
 $
The postmeasurement states $\{I_{A}\otimes M_{B}\otimes I_{C}|\phi_{i}\rangle$, $i=1$, $\cdots$, $2(n_{2}+n_{3}-1)-n_{1}\}$ should be mutually orthogonal. Consider the states $|\phi_{i}\rangle$ and $|\phi_{i+2(n_{1}-1)}\rangle$, $i=1,\cdots,n_{1}-1$, we have $\langle0-i|I_{A}|0\rangle \langle0|M_{B}^{\dagger}M_{B}|i\rangle \langle i|I_{C}|0-i\rangle=0$, that is, $\langle 0|M_{B}^{\dagger}M_{B}|i\rangle=0$. Hence, $b_{0i}=b_{i0}=0$, $i=1,\cdots,n_{1}-1$.

For the states $|\phi_{i+2(n_{1}-1)}\rangle$ and $|\phi_{j+2(n_{1}-1)}\rangle$, $i\neq j$, $i,j=1,2,\cdots,n_{1}-1$, we have $\langle0|I_{A}|0\rangle \langle i|M_{B}^{\dagger}M_{B}|j\rangle \langle0-i|I_{C}|0-j\rangle=0$, that is, $\langle i|M_{B}^{\dagger}M_{B}|j\rangle=0$. Hence, $b_{ij}=b_{ji}=0$, $i\neq j$, $i,j=1,2,\cdots,n_{1}-1$.

For the states $|\phi_{i+2(n_{1}-1)+(n_{2}-n_{1})}\rangle$ and $|\phi_{j+2(n_{1}-1)+(n_{2}-n_{1})}\rangle, i\neq j$, $i,j=n_{1},\cdots,n_{2}-1$, we have $\langle0-1|I_{A}|0-1\rangle \langle i|M_{B}^{\dagger}M_{B}|j\rangle \langle0|I_{C}|0\rangle=0$, that is, $\langle i|M_{B}^{\dagger}M_{B}|j\rangle=0$. Hence, $b_{ij}=b_{ji}=0$, $i\neq j$, $i,j=n_{1},\cdots,n_{2}-1$.

For the states $|\phi_{i+2(n_{1}-1)}\rangle$ and $|\phi_{j+2(n_{1}-1)+(n_{2}-n_{1})}\rangle$, $i=1,\cdots,n_{1}-1$; $j=n_{1},\cdots,n_{2}-1$, we have $\langle0|I_{A}|0-1\rangle \langle i|M_{B}^{\dagger}M_{B}|j\rangle \langle 0-i|I_{C}|0\rangle=0$, i.e., $\langle i|M_{B}^{\dagger}M_{B}|j\rangle=0$. Hence, $b_{ij}=b_{ji}=0$ for $i=1,\cdots,n_{1}-1$ and $j=n_{1},\cdots,n_{2}-1$.

For the states $|\phi_{1+(n_{1}-1)}\rangle$ and $|\phi_{j+2(n_{1}-1)+(n_{2}-n_{1})}\rangle, j=n_{1},\cdots,n_{2}-1$, we have $\langle1|I_{A}|0-1\rangle \langle 0-1|M_{B}^{\dagger}M_{B}|j\rangle \langle 0|I_{C}|0\rangle=0$, i.e., $\langle 0-1|M_{B}^{\dagger}M_{B}|j\rangle=0$. Therefore, $b_{0j}=b_{j0}=b_{1j}=0$, $j=n_{1},\cdots,n_{2}-1$.

For the states $|\phi_{i+(n_{1}-1)}\rangle$ and $|\phi_{2(n_{2}+n_{3}-1)-n_{1}}\rangle$, $i=1,\cdots,n_{1}-1$, we have $\langle i|I_{A}|0+\cdots+(n_{1}-1)\rangle \langle0-i|M_{B}^{\dagger}M_{B}|0+\cdots+(n_{2}-1)\rangle \langle 0|I_{C}|0+\cdots+(n_{3}-1)\rangle=0$, that is, $\langle 0-i|M_{B}^{\dagger}M_{B}|0+\cdots+(n_{2}-1)\rangle=0$. Hence, $b_{00}=b_{ii}$, $i=1,\cdots,n_{1}-1$. And for the states $|\phi_{j+2(n_{1}-1)+3(n_{2}-n_{1})}\rangle$ and $|\phi_{2(n_{2}+n_{3}-1)-n_{1}}\rangle$, $j=n_{1},\cdots,n_{2}-1$, we have $\langle 0|I_{A}|0+\cdots+(n_{1}-1)\rangle \langle(j-1)-j|M_{B}^{\dagger}M_{B}|0+\cdots+(n_{2}-1)\rangle \langle m|I_{C}|0+\cdots+(n_{3}-1)\rangle=0$, i.e., $\langle (j-1)-j|M_{B}^{\dagger}M_{B}|0+\cdots+(n_{2}-1)\rangle=0$. Hence, $b_{j-1,j-1}=b_{jj}$, $j=n_{1},\cdots,n_{2}-1$. Thus all diagonal elements of $M_{B}^{\dagger}M_{B}$ are equal.

Therefore, all of Bob's POVM elements $M_{B}^{\dagger}M_{B}$ are proportional to the identity operator and Bob cannot start with a nontrivial  measurement.

As for Charlie, we write the POVM elements $M_{C}^{\dagger}M_{C}$ in the basis $\{|0\rangle,|1\rangle,$ $\cdots,|n_{3}-1\rangle\}_{C}$:

$$M_{C}^{\dagger}M_{C}=\left(
\begin{array}{cccccccc}
c_{00}&\cdots&c_{0,n_{1}-1}&c_{0n_{1}}&\cdots&c_{0,n_{2}-1}&\cdots&c_{0,n_{3}-1}\\
c_{10}&\cdots&c_{1,n_{1}-1}&c_{1n_{1}}&\cdots&c_{1,n_{2}-1}&\cdots&c_{1,n_{3}-1}\\
\vdots&      &\vdots  &\vdots&      &\vdots  &   &\vdots\\
c_{n_{1}-1,0}&\cdots&c_{n_{1}-1,n_{1}-1}&c_{n_{1}-1,n_{1}}&\cdots&c_{n_{1}-1,n_{2}-1}&\cdots&c_{n_{1}-1,n_{3}-1}\\
c_{n_{1}0}&\cdots&c_{n_{1},n_{1}-1}&c_{n_{1}n_{1}}&\cdots&c_{n_{1},n_{2}-1}&\cdots&c_{n_{1},n_{3}-1}\\
\vdots&      &\vdots  &\vdots&      &\vdots &       &\vdots\\
c_{n_{2}-1,0}&\cdots&c_{n_{2}-1,n_{1}-1}&c_{n_{2}-1,n_{1}}&\cdots&c_{n_{2}-1,n_{2}-1}&\cdots&c_{n_{2}-1,n_{3}-1}\\
c_{n_{2}0}&\cdots&c_{n_{2},n_{1}-1}&c_{n_{2}n_{1}}&\cdots&c_{n_{2},n_{2}-1}&\cdots&c_{n_{2},n_{3}-1}\\
\vdots&      &\vdots  &\vdots&     &\vdots  &       &\vdots\\
c_{n_{3}-1,0}&\cdots&c_{n_{3}-1,n_{1}-1}&c_{n_{3}-1,n_{1}}&\cdots&c_{n_{3}-1,n_{2}-1}&\cdots&c_{n_{3}-1,n_{3}-1}\\
\end{array}
\right),
 $$

The postmeasurement states~$\{I_{A}\otimes I_{B}\otimes M_{C}|\phi_{i}\rangle, i=1,\cdots,2(n_{2}+n_{3}-1)-n_{1}\}$ should be mutually orthogonal. Consider the states $|\phi_{i}\rangle$ and $|\phi_{i+(n_{1}-1)}\rangle$, $i=1,\cdots,n_{1}-1$, we have $\langle0-i|I_{A}|i\rangle \langle0|I_{B}|0-i\rangle \langle i|M_{C}^{\dagger}M_{C}|0\rangle=0,$ i.e., $\langle i|M_{C}^{\dagger}M_{C}|0\rangle=0$. Thus, $c_{0i}=c_{i0}=0$ for $i=1,\cdots,n_{1}-1$.

For the states $|\phi_{i}\rangle$ and~$|\phi_{j}\rangle$, $i\neq j$, $i,j=1,2,\cdots,n_{1}-1$, we have $\langle0-i|I_{A}|0-j\rangle \langle0|I_{B}|0\rangle \langle i|M_{C}^{\dagger}M_{C}|j\rangle=0$, that is, $\langle i|M_{C}^{\dagger}M_{C}|j\rangle=0$. So $c_{ij}=0$, $i\neq j$, $i,j=1,2,\cdots,n_{1}-1$.

For the states~$|\phi_{1+(n_{1}-1)}\rangle$ and $|\phi_{i+2(n_{1}-1)}\rangle$, $i=n_{1},\cdots,n_{2}-1$, we have $\langle1|I_{A}|0-1\rangle \langle0-1|I_{B}|0\rangle \langle 0|M_{C}^{\dagger}M_{C}|i\rangle=0$, that is, $\langle 0|M_{C}^{\dagger}M_{C}|i\rangle=0$. Thus, $c_{0i}=c_{i0}=0$ for $i=n_{1},\cdots,n_{2}-1$.

For the states $|\phi_{i}\rangle$ and $|\phi_{j+2(n_{1}-1)}\rangle, i=1,2,\cdots,n_{1}-1; j=n_{1},\cdots,n_{2}-1$, we have $\langle0-i|I_{A}|0-1\rangle \langle0|I_{B}|0\rangle \langle i|M_{C}^{\dagger}M_{C}|j\rangle=0$, that is, $\langle i|M_{C}^{\dagger}M_{C}|j\rangle=0$. Hence, $c_{ij}=c_{ji}=0$ for $i=1,2,\cdots,n_{1}-1$ and $j=n_{1},\cdots,n_{2}-1$.

For the states $|\phi_{i+2(n_{1}-1)}\rangle$ and $|\phi_{j+2(n_{1}-1)}\rangle$, $i\neq j$, $i,j=n_{1},\cdots,n_{2}-1$, we have $\langle0-1|I_{A}|0-1\rangle \langle0|I_{B}|0\rangle \langle i|M_{C}^{\dagger}M_{C}|j\rangle=0$, i.e., $\langle i|M_{C}^{\dagger}M_{C}|j\rangle=0$. So $c_{ij}=0$, $i\neq j$, $i,j=n_{1},\cdots,n_{2}-1$.

For the states~$|\phi_{1+(n_{1}-1)}\rangle$ and $|\phi_{i+2(n_{1}-1)+3(n_{2}-n_{1})}\rangle$, $i=n_{2},\cdots,n_{3}-1$, we have $\langle1|I_{A}|0-1\rangle \langle0-1|I_{B}|0\rangle \langle 0|M_{C}^{\dagger}M_{C}|i\rangle=0,$ i.e., $\langle 0|M_{C}^{\dagger}M_{C}|i\rangle=0$. Hence, $c_{0i}=c_{i0}=0$, $i=n_{2},\cdots,n_{3}-1$.

For the states~$|\phi_{i}\rangle$ and $|\phi_{j+2(n_{1}-1)+3(n_{2}-n_{1})}\rangle$, $i=1,2,\cdots,n_{1}-1$; $j=n_{2},\cdots,n_{3}-1$, we have $\langle0-i|I_{A}|0-1\rangle \langle0|I_{B}|0\rangle \langle i|M_{C}^{\dagger}$ $M_{C}|j\rangle=0,$ i.e., $\langle i|M_{C}^{\dagger}$ $M_{C}|j\rangle=0$. Therefore, $c_{ij}=c_{ji}=0$ for $i=1,2,\cdots,n_{1}-1$ and $j=n_{2},\cdots,n_{3}-1$.

For the states $|\phi_{i+2(n_{1}-1)}\rangle$ and $|\phi_{j+2(n_{1}-1)+3(n_{2}-n_{1})}\rangle$, $i=n_{1},\cdots,n_{2}-1$; $j=n_{2},\cdots,n_{3}-1$, we have $\langle0-1|I_{A}|0-1\rangle \langle0|I_{B}|0\rangle \langle i|M_{C}^{\dagger}M_{C}|j\rangle=0,$ i.e., $\langle i|M_{C}^{\dagger}M_{C}|j\rangle=0$. Hence, $c_{ij}=c_{ji}=0$ for $i=n_{1},\cdots,n_{2}-1$ and $j=n_{2},\cdots,n_{3}-1$.

For the states $|\phi_{i+2(n_{1}-1)+3(n_{2}-n_{1})}\rangle$ and $|\phi_{j+2(n_{1}-1)+3(n_{2}-n_{1})}\rangle$, $i\neq j$, $i,j=n_{2},\cdots,n_{3}-1$, we have $\langle0-1|I_{A}|0-1\rangle \langle0|I_{B}|0\rangle \langle i|M_{C}^{\dagger}M_{C}|j\rangle=0,$ that is, $\langle i|M_{C}^{\dagger}M_{C}|j\rangle=0$. Hence, $c_{ij}=0$, $i\neq j$, $i,j=n_{2},\cdots,n_{3}-1$.

For the states $|\phi_{i+2(n_{1}-1)}\rangle$, $i=1,2,\cdots,n_{1}-1$ and ~$|\phi_{2(n_{2}+n_{3}-1)-n_{1}}\rangle$, we have $\langle0|I_{A}|0+\cdots+(n_{1}-1)\rangle \langle i|I_{B}|0+\cdots+(n_{2}-1)\rangle \langle 0-i|M_{C}^{\dagger}M_{C}|0+\cdots+(n_{3}-1)\rangle=0,$ that is, $\langle 0-i|M_{C}^{\dagger}M_{C}|0+\cdots+(n_{3}-1)\rangle=0$. So $c_{00}=c_{ii}$, $i=1,\cdots,n_{1}-1$.

For the states $|\phi_{j+2(n_{1}-1)+2(n_{2}-n_{1})}\rangle$, $j=n_{1},\cdots,n_{2}-1$ and ~$|\phi_{2(n_{2}+n_{3}-1)-n_{1}}\rangle$, we have $\langle0|I_{A}|0+\cdots+(n_{1}-1)\rangle \langle m|I_{B}|0+\cdots+(n_{2}-1)\rangle \langle (j-1)-j|M_{C}^{\dagger}M_{C}|0+\cdots+(n_{3}-1)\rangle=0,$ that is, $\langle (j-1)-j|M_{C}^{\dagger}M_{C}|0+\cdots+(n_{3}-1)\rangle=0$. Hence, $c_{j-1,j-1}=c_{jj}$, $j=n_{1},\cdots,n_{2}-1$. And for the states $|\phi_{k-n_{1}+2n_{2}+n_{3}-2}\rangle$, $e=n_{2},\cdots,n_{3}-1$ and $|\phi_{2(n_{2}+n_{3}-1)-n_{1}}\rangle$, we have $\langle0|I_{A}|0+\cdots+(n_{1}-1)\rangle \langle m|I_{B}|0+\cdots+(n_{2}-1)\rangle \langle (k-1)-k|M_{C}^{\dagger}M_{C}|0+\cdots+(n_{3}-1)\rangle=0,$ i.e., $\langle (k-1)-k|M_{C}^{\dagger}M_{C}|0+\cdots+(n_{3}-1)\rangle=0$. Hence, $c_{k-1,k-1}=c_{kk}$, $k=n_{2},\cdots,n_{3}-1$. Therefore, all diagonal elements of $M_{C}^{\dagger}M_{C}$ are equal.

All of Charlie's POVM elements are proportional to the identity operator and Charlie cannot start with a nontrivial measurement.

To sum up, all of the three participants Alice, Bob and Charlie cannot start with a nontrivial measurement. Therefore, the $2~(n_{2}+n_{3}-1)-n_{1}$ states cannot be perfectly distinguished by LOCC. This completes the proof.\\

\noindent{\bf C. Proof of Theorem 3}

\begin{proof}
Similarly, since the quantum state $\{|\varphi\rangle_{it}=|x\rangle_{it}|y\rangle_{it}|z\rangle_{it}, i=1,2,\cdots, l_{t}\}$ is a set of locally indistinguishable orthogonal product states in $\mathcal{C}^{d_{t}}\otimes \mathcal{C}^{d_{t+1}}\otimes \mathcal{C}^{d_{t+2}}$ $(t=1,4,\cdots,n-3)$, for $|\phi_{i}\rangle$, $i=1,2,\cdots l_{1}$, by Theorem 1,  the first three participants cannot perform a nontrivial measurement. And so on, any of the next three participants  cannot perform a nontrivial measurement. As for the last participant, with the first two participants, by Theorem 1, from the states $|\phi_{i+l_{1}+\cdots+l_{n-3}}\rangle$, $i=1,2,\cdots l_{n}$, we can arrive the same conclusion. To sum up, we prove it in a similar way that all of the participants can only make a trivial measurement. Thus the~$l_{1}+l_{4}+\cdots+l_{n}$ quantum states in~$\mathcal{C}^{d_{1}}\otimes \mathcal{C}^{d_{2}}\otimes\cdots\otimes \mathcal{C}^{d_{n}}$ $(d_1,d_2,\cdots,d_n \geq3$, $n>6$, $n\equiv1~(mod~3)$) quantum system cannot be perfectly distinguished by LOCC. This completes the proof.
\end{proof}

\noindent{\bf D. Proof of Theorem 4}

\begin{proof}
Similarly, since the quantum state $\{|\varphi\rangle_{it}=|x\rangle_{it}|y\rangle_{it}|z\rangle_{it}$, $i=1,2,\cdots, l_{t}\}$ is nonlocal indistinguishable orthogonal product states in~$\mathcal{C}^{d_{t}}\otimes \mathcal{C}^{d_{t+1}}\otimes \mathcal{C}^{d_{t+2}}$ $(t=1,4,\cdots,n-4)$, for $|\phi_{i}\rangle$, $i=1,2,\cdots l_{1}$, by Theorem 1, the first three participants cannot perform a nontrivial measurement. And so on, any of the next three participants  cannot perform a nontrivial measurement. As for the last two participants, with the first participant, by Theorem 1, from the states $|\phi_{i+l_{1}+\cdots+l_{n-4}}\rangle, i=1,2,\cdots l_{n-1}$, we can arrive the same conclusion. To sum up, we prove it in a similar way that all of the participants can only make a trivial measurement. So the $l_{1}+l_{4}+\cdots+l_{n-1}$ quantum states in~$\mathcal{C}^{d_{1}}\otimes \mathcal{C}^{d_{2}}\otimes\cdots\otimes \mathcal{C}^{d_{n}}$ $(d_1,d_2,\cdots,d_n \geq3$, $n>6$, $n\equiv2~(mod~3))$ quantum system cannot be perfectly distinguished by LOCC. This completes the proof.
\end{proof}

\end{document}